\documentclass[10pt,twocolumn,twoside]{IEEEtran}
\usepackage{generic}
\usepackage{cite}
\usepackage{comment}

\usepackage{amsmath,amssymb,amsfonts,yhmath,amsthm}
\usepackage{cuted} 
\usepackage{tikz}
\usepackage{array}
\newcommand\undermat[2]{%
  \makebox[0pt][l]{$\smash{\underbrace{\phantom{%
\begin{matrix}#2\end{matrix}}}_{\text{$#1$}}}$}#2}
\usetikzlibrary{arrows.meta, positioning}
\usepackage{algorithmic}
\usepackage{graphicx}
\usepackage{algorithm,algorithmic}
\usepackage{hyperref}
\usepackage{textcomp}
\newtheorem{theorem}{Theorem}
\newtheorem{lemma}{Lemma}
\newtheorem{example}{Example}
\newtheorem{proposition}{Proposition}
\newtheorem{corollary}{Corollary}
\newtheorem{definition}{Definition}

\newtheorem{problem}{Problem}
\newtheorem{remark}{Remark}

\DeclareMathOperator{\diag}{diag}

\newcommand{\baike}[1]{\textcolor{black}{#1}}

\def\BibTeX{{\rm B\kern-.05em{\sc i\kern-.025em b}\kern-.08em
    T\kern-.1667em\lower.7ex\hbox{E}\kern-.125emX}}
\begin{document}
\title{A Dissipativity Approach to Analyzing \\ Composite Spreading Networks}
\author{Baike She and Matthew Hale$^*$
\thanks{* 
Baike She and Matthew Hale are with the School of Electrical and Computer Engineering,
        Georgia Institute of Technology, Atlanta, GA, 30318, USA. 
        {\tt\small bshe6@gatech.edu; mhale30@gatech.edu}. This work was funded by DARPA under grant no. HR00112220038.}
}
\maketitle

\begin{abstract}
The study of spreading processes often analyzes networks at different
resolutions, e.g., at the level of individuals or countries,
but it is not always clear how properties at one resolution can carry
over to another. Accordingly, in this work we use dissipativity theory
from control system analysis 
to characterize composite spreading networks that are comprised
by many interacting subnetworks. 
We first develop a method to represent spreading networks that have inputs and outputs. Then we define a composition operation for 
composing multiple spreading networks into a larger composite spreading network. 
Next, we
develop storage and supply rate functions that can be used to demonstrate that spreading dynamics are dissipative. We then derive conditions 
under which a composite spreading network will converge to a disease-free equilibrium as long as its constituent spreading networks are dissipative with respect to those storage and supply rate functions. 
To illustrate these results, 
we use simulations of an influenza outbreak in a primary school, and we show
that an outbreak can be prevented by 
decreasing the average interaction time between any pair of classes to less than $79\%$ of the original interaction time.
\end{abstract}


\section{Introduction}
\label{sec:introduction}
Disease spreading network models are powerful tools for analyzing virus transmission dynamics at various resolutions, ranging from individual to continent-wide scales. When studying the spread of disease across multiple communities, it is common to construct spreading networks to understand transmission within each community, while also considering how these communities interact within a larger network. For example, during a pandemic, restricting travel between countries can localize the spread within each country~\cite{parr2020traffic,march2021tracking}.
Thus, we can consider the spreading network across different countries as a composite network, with its constituent networks representing the spread within each country.
Therefore, it is essential not only to study entire spreading networks, but also to understand their behavior by analyzing the interactions between their constituent networks. In this work, we 
use dissipativity theory to do so. 

Dissipativity analysis provides tools for studying the input-output behaviors of dynamical systems. One key benefit of dissipativity theory is that it enables
the analysis of a composite system by analyzing its subsystems' input-output interactions~\cite{willems2007dissipative}. 
It has been used, for example, in past work to analyze cyber-physical systems~\cite{zakeri2022passivity} and in 
the design and verification of large-scale systems~\cite{arcak2022compositional}.
These and other works
demonstrate that when subsystems exhibit dissipative input-output behavior, the overall composite system can achieve stability under certain types of interconnections~\cite{hill2022dissipativity,schweidel2021compositional}.

In the context of modeling and analyzing disease spreading processes, dissipativity analysis has been used to model propagation and mitigation of competing and coexisting viruses within a 
network~\cite{lee2015passivity,lee2016adaptive}. 
However, despite its success in studying network spreading processes, dissipativity analysis has rarely been applied to composite spreading networks formed by smaller 
constituent spreading networks, which one often encounters in practice.
Hence, we aim to fill this gap by understanding how the interconnection of constituent spreading networks influences the overall spreading dynamics of the composite network.

In this work, we consider the spreading behavior of the network Susceptible-Infected-Susceptible ($SIS$) model~\cite{van2011n}. 
Our contributions are as follows:
\begin{itemize}
\item We develop a model of spreading networks with inputs and outputs, as well as
an operation to compose such networks
(Definitions~\ref{Def_Input_Vector},~\ref{Def_Input_Transmission_Matrix}, and~\ref{Def_Comp_M})
\item Using dissipativity theory, we define storage and supply rate functions to characterize the input-output behavior individual networks (Theorem~\ref{Thm: SIS_Dissipativity})
\item We use those network analyses to determine conditions under which disease spreading in a composite network SIS model will go to zero (Theorems~\ref{Thm-Stability-Composite-Network} and~\ref{Thm_Composite_Net_Individual})
\item We derive explicit conditions under which disease spreading cannot
be guaranteed to go to zero (Corollary~\ref{cor_comp_network_1})
\item We illustrate these results by showing that, in simulations of an influenza spread in a primary school, 
decreasing the average
interaction time between any pair of classes to less than
$79\%$ of the original interaction time can 
effectively prevent an outbreak on campus.
\end{itemize}



This paper is organized as follows. Section~\ref{Sec_Background} gives background and problem statements, and Section~\ref{Sec_SIS_Model} defines network $SIS$ models with inputs and outputs. Section~\ref{Sec_Dissipative} explores dissipative network $SIS$ dynamics, and Section~\ref{sec_composite_sis} investigates 
composite network $SIS$ models. Section~\ref{sec_application} applies our 
results to real-world networks, and Section~\ref{sec_conclusion} concludes.

\subsection*{Notation}
We use~$\mathbb{R}$ to denote the reals and~$\mathbb{N}$ to denote the naturals. 
Let $\underline{n}$ denote the index set $\{1,2,\dots,n\}$.
For a finite set $S$, we use $|S|$ to denote its cardinality. 
For a matrix $A\in\mathbb{R}^{n\times n}$, we use $A_{ij}$ to denote the $i^{th}j^{th}$ entry of $A$. 
We use $A_{i,:}$ for the~$i^{th}$ row of a matrix~$A$
and $A_{:,j}$ for its $j^{th}$ column. 
We use $\rho(A)$  and $\sigma(A)$ to represent the spectral radius and spectral abscissa of the matrix $A$, respectively. 
We use $A\preceq 0$ to denote that~$A$ is negative semidefinite 
and $A\succeq 0$ to that~$A$ is positive semidefinite. For two matrices $A,B\in\mathbb R^{n\times n}$, we use $A\geq B$ to mean $A_{ij}\geq B_{ij}$ for all~$i, j \in \underline{n}$, and 
we use $A>B$ to mean both
$A_{ij}\geq B_{ij}$ and $A\neq B$. 
For a collection of matrices $A^i\in\mathbb R^{n_i\times n_i}$ with $i\in\underline m$, we use $\diag\{A_1, \cdots, A_m\}$ to represent the block diagonal matrix whose the $i^{th}$
diagonal block is~$A^i$. 
For~$x\in \mathbb{R}^n$, we use $\diag(x)\in\mathbb{R}^{n\times n} $ to denote the diagonal matrix whose $i^{th}$ diagonal entry is $x_i$ for all $i \in\underline{n}$. 
We use $\|x\|$ to represent the $2$-norm of the vector~$x$.
For~$x, y\in\mathbb{R}^n$, we use $x>y$ (or $x\geq y$) to denote that  $x_i>y_i$ (or $x_i\geq y_i$) for all $i\in\underline n$. We use $\boldsymbol{0}$ and $\boldsymbol{1}$ to denote the zero vector and one vector with the corresponding dimension given by context. 
Let $[a,b]^n$ denote a closed cube and $(a,b)^n$ denote an open cube,  for any $a,b\in \mathbb{R}$.

\section{Background and Problem Formulation}
\label{Sec_Background}
In this section, we introduce the background of network $SIS$ models and formally state the problems of interest.

\subsection{Network $SIS$ Spreading Dynamics}
We consider an epidemic spreading process over $m\in\mathbb N$ different directed,
strongly connected networks. For $k \in \underline m$,
we represent the $k^{th}$ spreading network by
a directed graph $\mathcal{G}^k=\left(\mathcal{V}^k,\mathcal{E}^k\right)$ on~$n_k \in \mathbb{N}$
nodes,  where the node set $\mathcal{V}^k=\left\{ v^k_{1},\ldots,v^k_{n_k}\right\}$ represents
$n_k$ entities (e.g., individuals, communities, countries) 
and the edge set $\mathcal{E}^k\subseteq \mathcal{V}^k\times \mathcal{V}^k$ 
represents the epidemic spreading interactions between them. 
For the $k^{th}$ spreading network, the epidemic
spreading interactions over $n_k$ nodes are captured by the transmission
matrix $B^k=\left[\beta^k_{ij}\right]\in\mathbb R_{\geq 0}^{n_k\times n_k}$
for all $i,j\in \underline{n_k}$. 
A directed edge within the $k^{th}$ spreading network, say~$\left(v^k_{j},v^k_{i}\right)$,
indicates that node $j$ can
infect node $i$ with the transmission rate $\beta^k_{ij}$, for all $i,j\in \underline n_k$. 

We use a network Susceptible-Infected-Susceptible ($SIS$) model on each of the~$k$ graphs
to model viral spreading dynamics. 
Within the~$k^{th}$ $SIS$ model, 
we define the dynamics of node $i\in\underline{n_k}$ as
\begin{equation}
\label{Eq_x_sis_no_input}
    \dot x^k_i (t) = -\gamma^k_i x^k_i (t) + \big(1-x^k_i(t)\big) \sum_{j\in \mathcal{N}^k_i} \beta^k_{ij}x^k_{j}(t),
\end{equation}
where $x^k_i\in[0,1]$ is the infected proportion within 
the $i^{th}$ node of the $k^{th}$ network, for all $k\in\underline m$ and $i\in \underline {n_k}$. 
We use $\gamma^k_i>0$ to represent the recovery rate of the $i^{th}$ node of the $k^{th}$ network for all $i \in \underline {n_k}$ and $k\in\underline m$. 
Additionally, we use $\mathcal{N}^k_i$ 
to represent the in-neighbor set of node $i$ in the $k^{th}$ spreading network, such that each node in the in-neighbor set of the node~$i$ can transmit the virus to node $i$ directly.

For network $SIS$ dynamics over a graph $\mathcal{G}^k$, we write the compact form of 
the system dynamics over all nodes as
\begin{equation}
\label{Eq: SIS_Cpct}
  \dot x^k (t) = \big(-\Gamma^k + B^k - \diag(x^k)B^k\big) x^k(t),
\end{equation}
where $x^k\in [0,1]^{n_k}$ is the infection vector and $x^k_i$ denotes the infected proportion within the $i^{th}$ node of the $k^{th}$ network. 
We use a diagonal matrix $\Gamma^k\in\mathbb R_{\geq 0}^{n_k\times n_k}$ to denote the recovery matrix, where $[\Gamma^k]_{ii} = \gamma^k_i$ for all $i\in \underline {n_k}$ and $k\in \underline m$.

\subsection{Motivational Example}
\label{Sec_Motivational_Example}
We frame our problem around mitigating epidemic spread across interconnected $SIS$ network models spanning multiple regions. During a global outbreak, halting inter-regional interactions helps avoid large-scale outbreaks. As the spread starts to subside, a gradual lifting of restrictions is needed to resume normal activities. However, lifting restrictions too quickly can trigger new waves of infection. Therefore, the relaxation of measures must be carefully managed to prevent a resurgence due to new transmissions between previously isolated regions.

\begin{figure}
    \centering
    \includegraphics[trim={5cm 5.2cm 4cm 4.3cm}, clip, width=0.6\columnwidth]{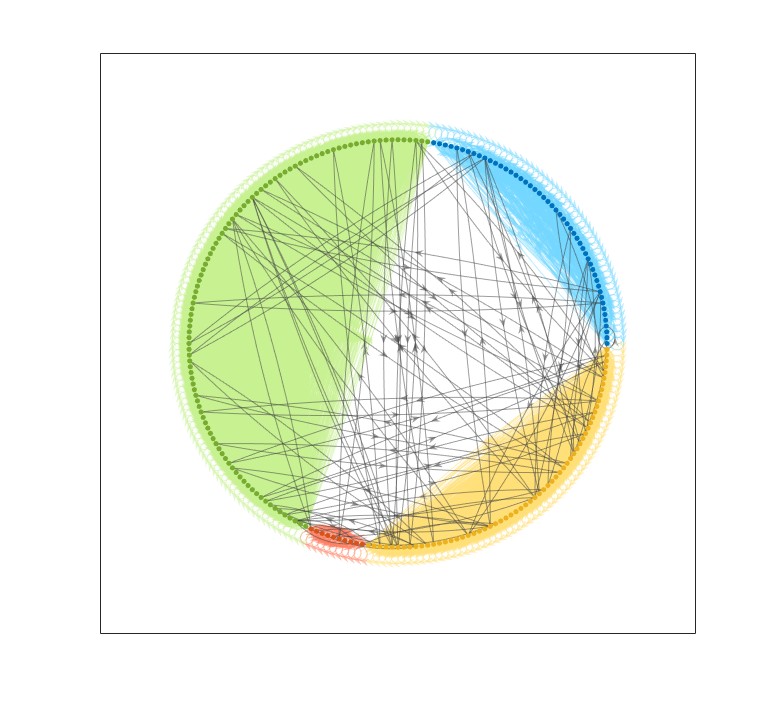}
    \caption{A composite spreading network is comprised of four subnetworks. Nodes in red, yellow, blue, and green around the perimeter represent populations within different networks, and edges of the same color (drawn in the interior) indicate transmission interactions within each network. Black edges connecting nodes from different networks (i.e., nodes of different colors) illustrate how these subnetworks combine to form the composite spreading network.}
    \label{fig:composite_network}
\vspace{-2ex}
\end{figure}
\begin{figure}
    \centering
    \includegraphics[trim={1.2cm 0.7cm 1.1cm 3cm}, clip, width=0.95\columnwidth]{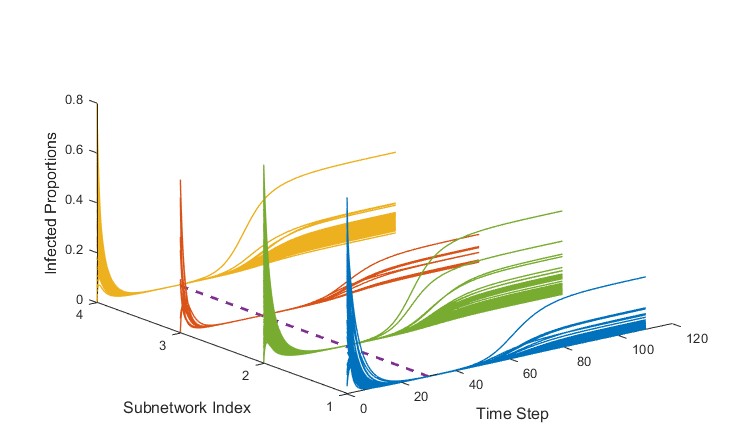}
    \caption{A plot of the sizes of the infected proportions for each network in Figure~\ref{fig:composite_network}, each of which is modeled using an~$SIS$ model. 
    Before the $30^{th}$ time step, there are no transmissions between nodes of different colors, meaning these networks are isolated.  However, at the $30^{th}$ time step (indicated by the dashed purple line), new transmission channels are introduced, leading to the presence of black edges in Figure~\ref{fig:composite_network}. These new transmission channels result in a resurgence of the outbreak across the entire composite network, despite the fact that each individual network had only a negligible portion of its population infected.}
    \label{fig:simulation_four_networks}
\end{figure}

Figure~\ref{fig:composite_network} illustrates four spreading networks governed by $SIS$ dynamics, where nodes and edges of the same color (red, yellow, green, and blue) represent four distinct, strongly connected spreading networks. 
Figure~\ref{fig:simulation_four_networks} shows the evolution of infections within
each region over time. 
From time step zero to~$30$, there are no interactions between nodes of different colors, which means that these networks are isolated. As shown in Figure~\ref{fig:simulation_four_networks}, the infections in these segregated networks begin to decrease toward zero, which may prompt policymakers to consider lifting restrictions as the infection subsides.
Now, suppose restrictions between these four areas are relaxed at time step $30$, introducing new transmission channels, which are represented by black edges between nodes of different colors in Figure~\ref{fig:simulation_four_networks}. This action re-establishes connections between the networks, leading to a resurgence of the outbreak after the $30^{th}$ time step, despite the fact that
infection levels were near zero within each network. Thus, interactions between communities can greatly increase epidemic spread, even when each community is experiencing negligible spread internally. 
This example highlights the critical need to carefully 
model outbreaks not only within communities but also across interactions among communities, and doing so
can enable the management of restrictions between previously isolated regions to prevent the recurrence of an outbreak.

\subsection{Problem Formulation}
Motivated by the example in Section~\ref{Sec_Motivational_Example}, we formulate five problems 
to analyze composite spreading networks. The scenario where the epidemic fades in the composite spreading network in Section~\ref{Sec_Motivational_Example}
corresponds to the dynamics converging to a stable, disease-free equilibrium. Therefore, in this work, we solve the following problems: 
\begin{problem}
\label{Problem-SIS-I/O}
Introduce a network $SIS$ model with inputs and outputs to capture both the internal transmission dynamics and the external transmission impact from other spreading networks. 
\end{problem}
\begin{problem}
\label{Problem-SIS-Dissipativity} 
Develop a supply rate function and a storage function to study the input-output behavior of the $SIS$ model proposed in Problem~\ref{Problem-SIS-I/O} and show that it is dissipative. 
\end{problem}
\begin{problem}
\label{Problem-SIS-Comp}   
Define an approach to construct composite $SIS$ spreading networks by composing multiple constituent spreading networks.
\end{problem}
\begin{problem}
\label{Problem-Comp-Dissipativity}
Quantify how the composition of constituent spreading networks impacts the overall behavior of the composite network through dissipativity analysis. 
\end{problem}
\begin{problem}
\label{Problem-Application}
Demonstrate the application of the developed theoretical results in decision-making for real-world pandemic mitigation strategies.
\end{problem}

We solve Problem~\ref{Problem-SIS-I/O} in Section~\ref{Sec_SIS_Model}, Problem~\ref{Problem-SIS-Dissipativity} in Section~\ref{Sec_Dissipative}, Problems~\ref{Problem-SIS-Comp} and~\ref{Problem-Comp-Dissipativity} in Section~\ref{sec_composite_sis}, and Problem~\ref{Problem-Application} in Section~\ref{sec_application}.

\section{Network \texorpdfstring{$SIS$}{SIS} Dynamics w/ Inputs and Outputs}
\label{Sec_SIS_Model}
We introduce a network $SIS$ spreading model with inputs and outputs in this section. As above, we consider the $k^{th}$ network, where $k \in \underline{m}$, from a total of $m \in \mathbb{N}$ spreading networks. In this setup, nodes in the $k^{th}$ network can be infected not only by other nodes within the same network but also by nodes from the $l^{th}$ network, for all $l \in \underline{m} \setminus \{k\}$.
Therefore, when there is an edge from node $j$ in network $\mathcal{G}^l$ to node $i$ in network $\mathcal{G}^k$, we use $\beta^{kl}_{ij}$ to denote the transmission rate, where $i \in \underline{n_k}$, $j \in \underline{n_l}$, and $k, l \in \underline{m}$ with $k \neq l$. For node~$i$ in the~$k^{th}$ spreading network, we use $\mathcal{N}^{kl}_i$ to represent its set of in-neighbors that are in the $l^{th}$ spreading network $\mathcal{G}^l$, for all $k, l \in \underline{m}$, with $k \neq l$.
The union of the in-neighbor sets of node $i$ in the $k^{th}$ spreading network, where $i \in \underline{n_k}$ and $k \in \underline{m}$, captures transmissions from all other spreading networks to node $i$ and is denoted 
\begin{equation}
\widehat{\mathcal{N}}^{k}_i:= \bigcup_{\substack{l \in \underline m \\ l \neq k}}{\mathcal{N}}^{kl}_i.
\end{equation}

Based on the network $SIS$ dynamics in~\eqref{Eq: SIS_Cpct},
for all $i\in \underline {n_k}$ and $k\in\underline m$,
we define the dynamics of the $i^{th}$ node in the $k^{th}$ spreading network with external inputs as
\begin{multline} \label{Eq_SIS_x_input}
\dot x^k_i (t) = -\gamma^k_i x^k_i (t) + \big(1-x^k_i(t)\big) \sum_{j\in \mathcal{N}^k_i}\beta^k_{ij}x^k_{j}(t)\\ 
+(1-x^k_i(t))
\sum_{j\in \widehat{\mathcal{N}}^{k}_i}\beta^{kl}_{ij}\underbrace{x^l_{j}(t)}_{\textnormal{Input}},
\end{multline}
where the external input $x^l_j\in [0,1]$ is the infected proportion of node $j$ in the $l^{th}$ network for all $j\in \widehat{\mathcal{N}}^{k}_i$ and 
$l\in\underline m$ with $k\neq l$. 
The rest of the notation in~\eqref{Eq_SIS_x_input} is defined as in~\eqref{Eq_x_sis_no_input}. 
The proposed dynamics in~\eqref{Eq_SIS_x_input} capture not only the transmission between nodes within the $k^{th}$ network $\mathcal{G}^k$, but also the transmissions from nodes of other networks through the inputs. 
In addition, the transmission mechanisms between networks follow the same dynamics as those governing the interactions between individual nodes within a single network.

Building upon~\eqref{Eq_SIS_x_input}, 
we define the input vector of the $k^{th}$ network. 

\begin{definition}[Input Vector $u^{k}$] \label{Def_Input_Vector}
We denote the input vector of the $k^{th}$ spreading network as $u^k(t)\in \mathbb R^{p_k}$, where ${p_k=\big|\cup_{i\in \underline n_k}\widehat{\mathcal{N}}^{k}_i\big|}$ represents the number of distinct inputs, i.e., total number of $p_k$ nodes in other networks that have direct interactions to any node in the $k^{th}$ network. The input vector $u^k(t)$ is concatenated from infected states from other spreading networks, namely  $\mathcal{G}^l$ with $l\in\underline m\setminus \{k\}$, and it is constructed as follows: 
\begin{enumerate}
    \item The infected states $x^l_j$, where $j \in \underline{n_l}$ and $l \in \underline{m}\setminus \{k\}$, in $u^{k}$ are first arranged in ascending order based on their superscript $l$, which indicates the network to which they belong;
    \item If two or more infected states share the same superscript, they are ordered by their subscripts in ascending order, i.e., by the node number $j$ in their corresponding network, where $j \in \underline{n_l}$ and $l \in \underline{m}\setminus \{k\}$. 
\end{enumerate}    
\end{definition}

\begin{figure}[ht]
    \centering \includegraphics[width=1\linewidth]{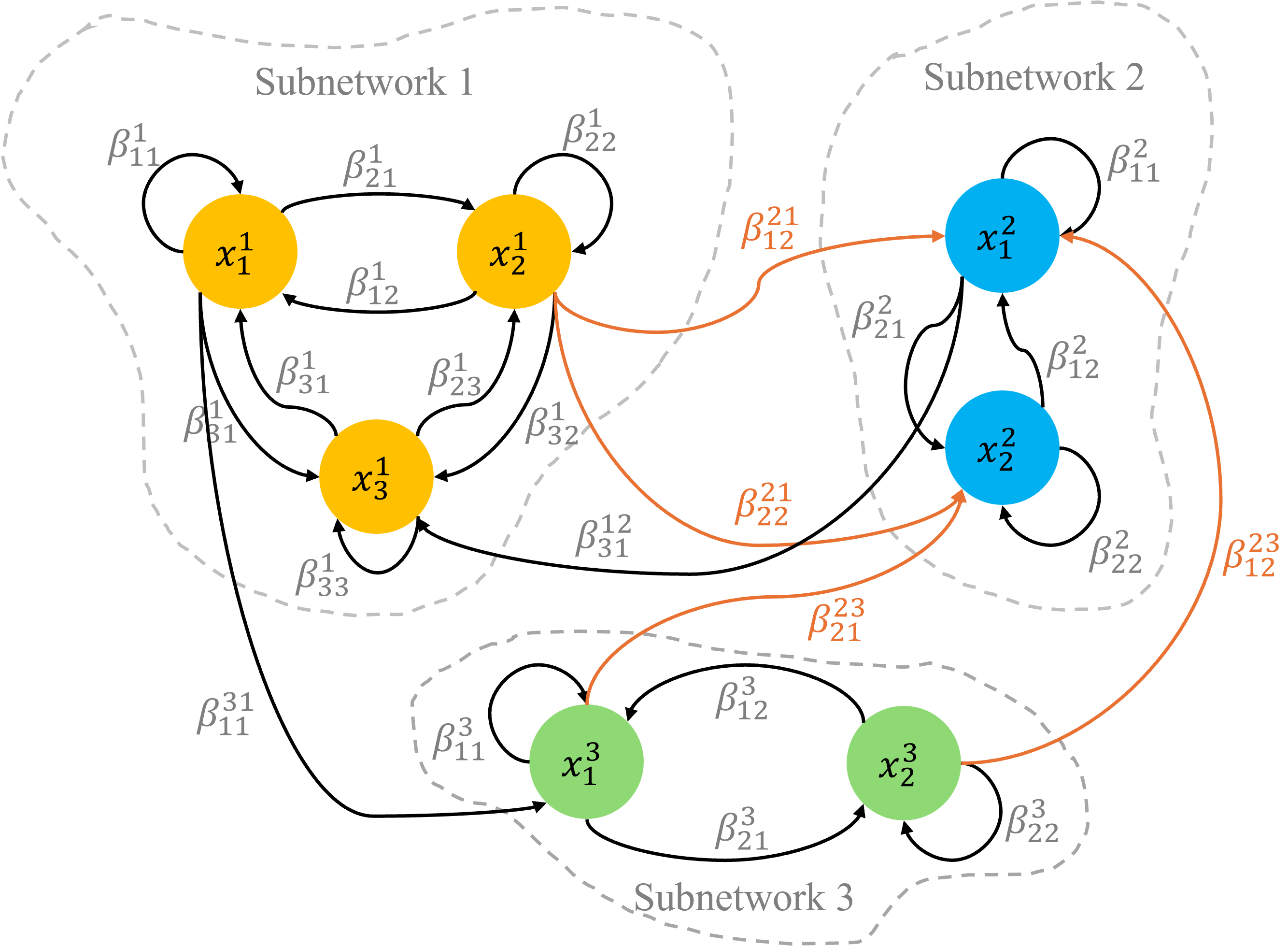}
    \caption{Composite Network of Three Subnetworks. Nodes with the same color are in the same subnetwork. We use orange edges to represent the transmissions from subnetworks 1 and 3 to subnetwork 2.}
\label{fig:Three_Subnet_Example}
\end{figure}

The next example illustrates the composition of three networks with inputs. 

\begin{example}
\label{Example_Input_Vector}
Figure~\ref{fig:Three_Subnet_Example} presents three spreading networks indexed over $k\in\underline 3$. 
We use black edges to represent transmission interactions within the same network and orange edges to represent transmission interactions across different spreading networks.
Consider Subnetwork~$2$ in Figure~\ref{fig:Three_Subnet_Example}. 
According to Definition~\ref{Def_Input_Vector}, 
we construct the input vector of subnetwork~$2$ as
\begin{align*}
   u^2 =[u^{2}_1, u^{2}_2, u^{2}_3]^\top=  [x^1_2, x^3_1, x^3_2]^\top.
\end{align*}
Here, we use $u^{2}_2$ to represent the input variable from node $1$ in Subnetwork~$3$ to Subnetwork~$2$, which is the input variable $x^3_1$. In addition, $u^2$ concatenates all inputs $x^1_2, x^3_1, x^3_2$
based on the rules in Definition~\ref{Def_Input_Vector}.
\end{example}

Definition~\ref{Def_Input_Vector} defines the input vector $u^k$. Next, we define the input transmission matrix to map the input variables in $u^k$ to the specific nodes in the $k^{th}$ spreading network, for all $k\in\underline m$.

\begin{definition}[Input Transmission Matrix] \label{Def_Input_Transmission_Matrix}
The input transmission matrix of the $k^{th}$ spreading network is defined as
$C^k \in \mathbb R^{n_k\times p_k}$. A non-zero entry takes
the form $C^k_{ij}= \beta^{ka}_{ib}$, where $a$ and $b$ are such that $x^a_b = u^{k}_j$, 
where~$u^k$ is from Definition~\ref{Def_Input_Vector}. 
If $x^k_i$ has no external inputs, then $C^k_{ij}=0$ 
for all $i\in\underline {m_k}$ and $j\in \underline {{p}_k}$. 
\end{definition}

\begin{remark}
We explain the meaning of the entry $C^k_{ij}= \beta^{ka}_{ib}$. 
The superscript $k$ in
$C^k_{ij}= \beta^{ka}_{ib}$ 
indicates that this term
captures transmissions into the~$k^{th}$ network
from other networks $\mathcal{G}^l$, where
$l\in\underline{m}\setminus \{k\}$.
The subscripts $i$ and $j$ in
$C^k_{ij}$ denote that it models the effect of 
the $j^{th}$ entry in the input vector $u^k$ upon the $i^{th}$ node in the $k^{th}$ network. 
According to Definition~\ref{Def_Input_Vector}, each entry in the input vector corresponds to an infected state in another spreading network, e.g., $u^k_j = x^a_b$, as illustrated in Example~\ref{Example_Input_Vector}. 
Therefore, we can determine the parameters $a$ and $b$ in $C^k_{ij}$ 
based on the entries of the input vector $u^k$. The remaining entries in $C^k$ 
correspond to terms that are absent from~$u^k$, which result from a lack of connection
between nodes, and these entries are filled with zeros. 
\end{remark}

Definition~\ref{Def_Input_Transmission_Matrix} constructs the input transmission matrix $C^k$ for the $k^{th}$ network $\mathcal{G}^k$ by using (i) the input vector $u^k$ and (ii) the transmission rates from other networks $\mathcal{G}^l$ (with~$l \neq k$) to the $k^{th}$ network. Using Definitions~\ref{Def_Input_Vector} and~\ref{Def_Input_Transmission_Matrix}, we derive a compact form for the dynamics of the $k^{th}$ spreading network with external inputs, namely
\begin{multline}
\label{Eq: SIS_Input}
  \dot x^k (t) = \big(-\Gamma^k + B^k - \diag(x^k)B^k\big) x^k(t)   \\ 
   + \big(C^k - \diag(x^k)C^k\big) u^k (t),
\end{multline}
where $x^k$, $B^k$, and $\Gamma^k$ are the state vector, transmission matrix, and recovery matrix of the $k^{th}$ spreading network $\mathcal{G}^k$, defined in~\eqref{Eq: SIS_Cpct}.

The $SIS$ model with inputs in~\eqref{Eq: SIS_Input} solves Problem~\ref{Problem-SIS-I/O}. In the next section, we leverage dissipativity analysis to study 
its input-output behavior. Before that, we given an example showing
the construction of such an $SIS$ model.

\begin{example}
We illustrate how to use Definition~\ref{Def_Input_Transmission_Matrix} to construct the input matrix for the $2^{nd}$ subnetwork in Figure~\ref{fig:Three_Subnet_Example}. 
Figure~\ref{fig:Three_Subnet_Example} shows that Subnetwork~$2$ has two nodes in it and three inputs from three total nodes across Subnetworks~$1$ and~$3$. According to Definition~\ref{Def_Input_Vector}, 
 the dimension of the input matrix is given by $C^2\in\mathbb R^{n_2\times p_2}$, where $n_2 = 2$ and $p_2 = 3$.
Based on Definition~\ref{Def_Input_Transmission_Matrix},
we use the transmission rates in Figure~\ref{fig:Three_Subnet_Example} to populate the entries in $C^2$ as
\begin{equation}
C^2 = \left[\begin{array}{ccc}
     \beta^{21}_{12} & 0 &   \beta^{23}_{12}   \\[2pt]
     \beta^{21}_{22} &      \beta^{23}_{21}          & 0 
    \end{array}\right].
\end{equation}
We explain each entry in $C^2$ by showing how to construct the entry $C^2_{21}$. According to Definition~\ref{Def_Input_Transmission_Matrix},
a non-zero entry takes the form $C^k_{ij} = \beta^{ka}_{ib}$.
We are considering the spread to a node in Subnetwork~$2$, which sets~$k = 2$,
and we are considering the spread to the second node in that network, which
sets~$i = 2$. 
Therefore, we have that $C^2_{2j} = \beta^{2a}_{2b}$. 
Next, we determine $a$ and $b$. Example~\ref{Example_Input_Vector} shows that the input vector $u^2= [x^1_2, x^3_1, x^3_2]^\top$. 
Because we consider~$C^2_{21}$, we have~$j=1$, and we see 
that~$u^{k=2}_{j=1} = x^{a=1}_{b=2}$. Therefore, $C^2_{21} = \beta^{21}_{22}$.

By populating the entirety of~$C^2$, 
we write the spreading dynamics from external nodes for Subnetwork~$2$ as
\begin{multline}
    \big(I-\diag(x^2)\big) C^2u^2 \\
    = 
    \left(\left[\begin{array}{cc}
        1 & 0\\
        0 & 1
    \end{array}\right]
    -
     \left[\begin{array}{cc}
         x^2_1 & 0 \\  
         0 & x^2_2 
     \end{array}\right]\right)
     \left[
     \begin{array}{ccc}
    \beta^{21}_{12} & 0 &   \beta^{23}_{12}   \\
     \beta^{21}_{22} &      \beta^{23}_{21}          & 0 
    \end{array}\right]
     \left[\begin{array}{c}
         x^1_2 \\ 
         x^3_1 \\ 
         x^3_2
     \end{array}\right] \\ 
     \!\!\!=\!  (1 \!-\! x^2_1) (\beta^{21}_{12}x^1_2 + \beta^{23}_{12}x^3_2) + (1 \!-\! x^2_2) ( \beta^{21}_{22}x^1_2 + \beta^{23}_{21}x^3_1),
\end{multline}
where $x^2\in[0,1]^2$ is the infection state vector of Subnetwork~$2$.
\end{example}
\section{Dissipative Network \texorpdfstring{$SIS$}{SIS} Systems}
\label{Sec_Dissipative}
In this section, we begin by defining dissipative $SIS$ spreading networks with inputs and outputs. Then, we introduce a supply rate function and a storage function to analyze the input-output behavior of the $SIS$ dynamics in~\eqref{Eq: SIS_Input}. Next, we examine the conditions under which the dynamics 
are dissipative with respect to these functions, 
which provides insight into the spreading networks' equilibrium properties.


We define dissipative spreading networks according to standard
definitions in dissipativity theory~\cite{willems1972dissipative}.
Dissipativity characterizes dynamical
systems broadly by how their inputs and outputs correlate. We thus begin
by rewriting
each network's $SIS$ dynamics in the typical form needed for dissipativity
analysis. 

We write the dynamics of the $k^{th}$ spreading network $\mathcal{G}^k$ in~\eqref{Eq: SIS_Input} as
\begin{align}
        \dot x^k(t) = f^k\big(x^k(t),u^k(t)\big), \ \ f^k(0, 0) = \boldsymbol{0} \label{eq:dyn}, 
\end{align}
where 
${f^k: \mathbb R^{n_k} \times \mathbb R^{p_k} \rightarrow \mathbb R^{n_k}}$ 
is equal to the right-hand side of~\eqref{Eq: SIS_Input}. 
Then the condition $f^k(0, 0) = \boldsymbol{0}$ captures the disease-free equilibrium~\cite{van2011n} of the $k^{th}$ network's $SIS$ dynamics in~\eqref{Eq: SIS_Input}.
We further define the output mapping 
of the $k^{th}$ network in~\eqref{Eq: SIS_Input} as 
\begin{align}
\label{eq:output}
    y^k(t) = h^k\big(x^k(t), u^k(t)\big) = x^k(t), \ \ y^k(0, 0)=\boldsymbol{0}, 
\end{align}
where~$h^k: \mathbb R^{n_k} \times \mathbb R^{p_k} \rightarrow \mathbb R^{n_k}$. 
We assume that the states of all nodes in the $k^{th}$ network can be monitored,
and thus the output vector $y^k = x^k$ includes all infection states. 

We can now define dissipative network $SIS$ dynamics.
\begin{definition}[Dissipative Network $SIS$ Spreading Dynamics]
\label{Def: Dissipativity}
Consider a supply rate function $S^k : \mathbb R^{p_k} \times \mathbb R^{n_k} \rightarrow \mathbb R$ for 
the $k^{th}$ spreading network in~\eqref{eq:dyn} and~\eqref{eq:output}. 
Then $k^{th}$ spreading network is dissipative with respect to~$S^k$ 
if there is a continuously differentiable storage function ${V^k : \mathbb R^{n_k} \rightarrow \mathbb R_{\geq 0}}$, where (i) $V^k(0) = 0$, (ii) $V^k(x^k)\geq 0$ for all~$x^k \in \mathbb{R}^{n_k}$, 
and~(iii)
\begin{align}
\label{eq_dissipativity_c}
    \nabla V^k(x^k)^\top f^k(x^k,u^k) \leq S^k(u^k,y^k),
\end{align}
for all $x^k, y^k \in \mathbb R^{n_k}$ and $u^k\in \mathbb R^{p_k}$.
\end{definition}

\begin{remark}
Following~\cite{arcak2016networks}, we define dissipativity in terms of a supply rate function and 
a storage function. In the network $SIS$ dynamics in~\eqref{Eq: SIS_Input}, the storage function $V^k$ depends on the level of infections within the $k^{th}$ network, which is~$x^k$,
while the supply rate function depends in part on infections from other networks $\mathcal{G}^l$ 
with~$l \neq k$. 
According to~\eqref{eq_dissipativity_c}, the change in the storage function~$V^k$ along the trajectory of the dynamics~$f^k$, under the influence of external infections represented by~$u^k$, cannot exceed the value of the supply rate function~$S^k$. This supply rate reflects both the internal infections within the $k^{th}$ network and the impact of transmissions from other networks. Therefore, dissipativity analysis allows us to study the spread within the network $\mathcal{G}^k$ in a way
that accounts for its interactions with other networks $\mathcal{G}^l$. 
\end{remark}

After defining dissipative network $SIS$ dynamics in~Definition~\ref{Def: Dissipativity}, we study the conditions under which the model in~\eqref{Eq: SIS_Input} is dissipative. Definition~\ref{Def: Dissipativity} indicates that, in order to show the 
network $SIS$ model is dissipative, we need to find suitable supply rate and storage functions.

\begin{definition} [Supply Rate Function for Network $SIS$ Dynamics]
\label{Def_Supply_Rate}
We define the supply rate of the $k^{th}$ network $SIS$ spreading dynamics in~\eqref{eq:dyn}-\eqref{eq:output} as
\begin{align}
\label{Eq_Supply_Rate}
\nonumber
S^k(u^k,y^k) &= y^k(t)^\top C^k u^k(t) + y^k(t)^\top(-\Gamma^k+ B^k) y^k(t)\\ 
  &= 
\begin{bmatrix}
u^k\\
y^k
\end{bmatrix}^\top 
\begin{bmatrix}
\boldsymbol{0} & \frac{1}{2}{C^k}^{\top}\\
\frac{1}{2}C^k & -\Gamma^k+ B^k
\end{bmatrix}
\begin{bmatrix}
u^k \\ y^k
\end{bmatrix},
\end{align}
for all $k\in\underline m$.    
\end{definition}
\begin{definition}[Storage Function for Network $SIS$ Dynamics]
\label{Def_Storage_Function}
   We define the storage function for the $k^{th}$ spreading network in~\eqref{eq:dyn}-\eqref{eq:output} as
    \begin{align}
    \label{Eq_Storage_Function}
    V^k(x^k) = \frac{1}{2}\|x^k\|^2
    \end{align}
for all $k\in\underline m$.
\end{definition}

\begin{remark}
We define both the supply rate and storage functions using quadratic forms. 
Recall from Definition~\ref{Def_Input_Vector} that $u^k$ is an input vector that captures all the external infected states that can directly infect the nodes that receive inputs in the $k^{th}$ spreading network. Further, according to~\eqref{eq:output}, 
the output~$y^k = x^k$ indicates that we can observe the infected states of all nodes. Hence, 
the first term of the supply rate function in~\eqref{Eq_Supply_Rate} represents a weighted product between the infected states at the nodes receiving inputs in the $k^{th}$ network and the corresponding inputs from other spreading networks $\mathcal{G}^l$ with~$l \neq k$. 
The second term represents a weighted sum of all infected states within the $k^{th}$ spreading network. The supply rate function $S^k$ is indeterminate because the first term is always non-negative, while the second term can be negative. 

The storage function $V^k$ works analogously to the kinetic energy function in a physical system, where a higher proportion of infection in the $k^{th}$ spreading network will generate a higher value of the storage function, indicating higher ``energy''. 
\end{remark}

Building on the supply rate and storage functions, we first explore the conditions under which the spreading network $\mathcal{G}^k$ is dissipative for all $k \in \underline{m}$.
\begin{theorem} \label{Thm: SIS_Dissipativity}
For all $k\in \underline m$, the $k^{th}$ spreading network  in~\eqref{Eq: SIS_Input} is dissipative with respect to the storage function $V^k(x^k) = \frac{1}{2}\| x^k\|^2$ and the supply rate function $S^k(u^k,y^k) = y^k(t)^\top C^k u^k(t) + y^k(t)^\top(-\Gamma^k+ B^k) y^k(t)$. 
\end{theorem}
\emph{Proof: } See Appendix~\ref{app:t1proof}. \hfill $\blacksquare$

Theorem~\ref{Thm: SIS_Dissipativity}  characterizes the dissipativity of the $SIS$ 
spreading network $\mathcal{G}^k$ with external inputs coming from other spreading networks $\mathcal{G}^l$ with~$l \neq k$, which addresses Problem~\ref{Problem-SIS-Dissipativity}.
The second term of the supply rate $S^k(u^k,y^k)$ is ${y^k}^\top (-\Gamma^k+B^k) y^k$, which
is a quadratic form associated with the matrix $(-\Gamma^k+B^k)$. For the network $SIS$ dynamics without external inputs given in~\eqref{Eq: SIS_Cpct}, the spectral abscissa of the matrix $-\Gamma^k+B^k$  determines the properties of its
equilibria.
\begin{proposition}[\cite{fall2007epidemiological}] \label{Prop: SIS_Equi}
The network $SIS$ dynamics for the strongly connected spreading network $\mathcal{G}^k$ in~\eqref{Eq: SIS_Cpct} has
\begin{enumerate}
    \item the unique globally stable disease-free equilibrium at $x^k = \mathbf{0}$ if and only if $\sigma(-\Gamma^k+B^k)\leq 0$ (equivalently, $\rho ({\Gamma^k}^{-1}B^k) \leq 1$) for all $k\in\underline m$
    \item the unique globally stable endemic equilibrium at $x^k \in (0,1)^{n_k} $ if and only if $\sigma(-\Gamma^k+B^k)> 0$ (equivalently $\rho ({\Gamma^k}^{-1}B^k) > 1$). Meanwhile, the disease-free equilibrium at $x^k= \mathbf{0}$ is unstable.
\end{enumerate}
\end{proposition}

In dissipativity analysis, when the supply rate is negative, the system is dissipating more energy than it is receiving. We further analyze the properties of the supply rate in~\eqref{Eq_Supply_Rate} through the following result.

\begin{lemma} \label{Lem: supply_rate}
Consider an undirected spreading network $\mathcal{G}^k$ with $k\in\underline m$. 
The supply rate $S^k(u^k,y^k)= {y^k}^\top C^k u^k + {y^k}^\top (-\Gamma^k +B^k) g^k$ is negative only if the spreading network $\mathcal{G}^{k}$ has the globally stable disease-free equilibrium at $x^k = \boldsymbol{0}$. 
\end{lemma}
\emph{Proof: } See Appendix~\ref{app:l1proof}. \hfill $\blacksquare$

\begin{remark}
If the spreading network $\mathcal{G}^k$ is directed, then the supply rate~$S^k(u^k, y^k)$ in~\eqref{Eq: SIS_Cpct} can be negative, even if the network has an unstable disease-free equilibrium,. In the case of directed spreading networks, the sign of the term ${y^k}^\top(-\Gamma^k + B^k){y^k}$ does not depend on the eigenvalues of $(-\Gamma^k + B^k)$. Therefore, Lemma~\ref{Lem: supply_rate} only applies to undirected spreading networks.
\end{remark}


\section{Composite \texorpdfstring{$SIS$}{SIS} Spreading Network} \label{sec_composite_sis}
After introducing the dissipative network $SIS$ spreading dynamics, and the corresponding supply rate and storage functions in Section~\ref{Sec_Dissipative}, we explore the properties of a composite spreading network composed of dissipative constituent networks. First, we define the process of composing the $m$ spreading networks $\mathcal{G}^{k}$, where $k \in \underline{m}$, into a single composite network.

In Section~\ref{Sec_Dissipative}, the input-output perspective of the spreading network in~\eqref{Eq: SIS_Input}
provided a way to characterize the dynamics of interconnected spreading networks. 
Consider $m\in\mathbb N$ spreading networks $\mathcal{G}^k$, where the dynamics 
of the $k^{th}$ spreading network $\mathcal{G}^k$ are modeled
by~\eqref{eq:dyn}-\eqref{eq:output}. 
We further define a static matrix $M$ that captures the  coupling of these spreading networks. 
As indicated by~\eqref{Eq: SIS_Input}, 
the input $u^k$ of the $k^{th}$ network 
depends on the output $y^l$ of other networks $\mathcal{G}^l$ with~$l \neq k$. 
According to~\eqref{eq:output}, we have $y^k = x^k$ for all $k\in\underline m$. Therefore, by concatenating all the input vectors $u^k\in\mathbb R^{p_k}$ and all the output vectors $y^k\in\mathbb R^{n_k}$ for all $k,l\in\underline{m}$, we have 
\begin{align*}
  u = [{u^1}^\top, \dots, {u^m}^\top]^\top, \textnormal{\ and \ } y = [{y^1}^\top, \dots, {y^m}^\top]^\top,
\end{align*}
where $u\in\mathbb R^{\mathbf{p}}$ and $y\in\mathbb R^{\mathbf{n}}$, and $\mathbf{p} = \sum_{k=1}^m{p_k}$, $\mathbf{n} = \sum_{k=1}^m{n_k}$. 

\begin{definition} [Composition Matrix $M$] \label{Def_Comp_M}
The composition matrix $M\in \mathbb{R}^{\mathbf{p}\times \mathbf{n}}$ satisfies $u = My$.
The composition matrix $M$ is constructed through the following rules:
\begin{itemize}
    \item The~$j^{th}$ row of~$M$ is 
\begin{align*}
    \mathbb{R}^{1 \times \mathbf{n}} \ni M_{j,:} = [M^{j1},\cdots,M^{jm}],
\end{align*}
    where $M^{jk}\in \mathbb R^{1\times n_k}$
    \item If $u_j = x^k_b$, then for the $j^{th}$ row of $M$, we have that $M^{jk}_b = 1$, where $M^{jk}_b$ denotes the~$b^{th}$ entry of the vector~$M^{jk}$
    \item The rest of the entries in $M$ are set to zero.
\end{itemize}
\end{definition}

According to the composition, we define the state vector of the composite spreading network $\mathcal{G}$ as 
\begin{align*}
    x = [{x^1}^\top,\dots, {x^m}^\top]^\top,
\end{align*}
where $x\in\mathbb R^{\mathbf{n}}$.
Definition~\ref{Def_Comp_M} introduces a method for constructing a composite matrix to connect the outputs of the constituent spreading networks to their 
corresponding inputs, which generates a composite spreading network. 
This capability solves Problem~\ref{Problem-SIS-Comp}, and
we illustrate it in the following example. 

\begin{example}
Consider the composite spreading network that is comprised by 
the three spreading networks in Figure~\ref{fig:Three_Subnet_Example}. According to 
Definition~\ref{Def_Input_Vector},
    the input vectors of the three networks are given by  
\begin{align*}
    u^1=x^2_1, \quad 
    u^2 = [x^1_2,x^3_1, x^3_2]^\top, \quad
    u^3 = x^1_1.
    \end{align*}
Therefore, the concatenated input vector of the overall network for composition is given by 
$u= [x^2_1, x^1_2,x^3_1, x^3_2, x^1_1]^\top$.
Further, the concatenated output vector of the three spreading networks is given by
$y=
[{y^1}^\top, {y^2}^\top, {y^3}^\top]^\top$, 
which is 
ordered ascending first with respect to the index of the subnetwork, i.e., the superscript, then ordered ascending with respect to the node index in each subnetwork, i.e.,  the subscript. Based on Definition~\ref{Def_Comp_M}, we have 
\begin{align*}
\underbrace{\begin{bmatrix}
  x^2_1\\ x^1_2\\x^3_1\\ x^3_2\\ x^1_1
\end{bmatrix}}_{u}
=
\underbrace{\begin{bmatrix}
      0 & 0& 1& 0& 0& 0& 0 \\
      0 & 1& 0& 0& 0 & 0& 0\\
      0 & 0& 0& 0& 0 & 1& 0\\
      0 & 0& 0& 0& 0 & 0& 1\\
      1 & 0& 0& 0& 0 & 0& 0\\
  \end{bmatrix}}_{M}
\underbrace{
\begin{bmatrix}
x^1_1\\ x^1_2\\ x^2_1\\x^2_2 \\ x^2_3 \\x^3_1\\x^3_2
\end{bmatrix}}_{y}.
\end{align*}
To highlight the properties of this composition, 
we examine the mapping from $x^{31}$ in $y$ to $x^{31}$ in $u$. Based on Definition~\ref{Def_Comp_M}, $u_3 = x^{k=3}_{b=1}$  indicates that we will consider  the $3^{rd}$ row in $M$, which is denoted $M_{3,:}$. Further, $x^3_1$ is the first node in the third Subnetwork, and thus
$M^{33}_1 = 1$. The rest of the entries in $M_{3,:}$ are zeros.
Hence, we have that
\begin{align*}
M_{3,:} =\left [
\begin{array}{rr|rrr|rr}
\undermat{M^{31}}{0 & 0 } & \undermat{M^{32}}{0 & 0 & 0} & \undermat{M^{33}}{1 & 0 } \\
\end{array}
\right ].    
\end{align*}
\vspace{1ex}
\end{example}

\subsection{Dissipativity Analysis of Composite Spreading Networks}
In this section, we study a composite network of $m$ constituent spreading networks $\mathcal{G}^k$, as described by~\eqref{eq:dyn}-\eqref{eq:output}, where $k \in \underline{m}$. We denote the composite spreading network as $\mathcal{G}_C$.
The composite spreading dynamics of $\mathcal{G}_C$ 
are defined from the dynamics of $\mathcal{G}^k$ by the composition matrix $M$ in~Definition~\ref{Def_Comp_M}. In particular, we can represent  
the dynamics of $\mathcal{G}_C$ by
\begin{multline} \label{Eq_SIS_Composite_Network_w_M}
  \dot {x} (t) = f_c(x) := \big({-\Gamma} + B - \diag(x)B\big) x(t) \\ 
  + \big(C - \diag(x)C\big) Mx(t),
\end{multline}
where $x=[{x^1}^\top, \cdots, {x^n}^\top]^\top \in \mathbb R^{\mathbf {n}}$.
Based on the composition rules, the transmission, recovery, and input matrices are given by $B = \diag\{B^1,\dots, B^m\}$, $\Gamma = \diag\{\Gamma^1,\dots, \Gamma^m\}$, and $C = \diag\{C^1,\dots, C^m\}$, respectively, where $B^k$, $\Gamma^k$, and $C^k$ are defined in~\eqref{Eq: SIS_Input}. We also use the conditions $u=My=Mx$ in Definition~\ref{Def_Comp_M}.
Unlike the strongly connected constituent spreading networks $\mathcal{G}^k$, the composite network~$\mathcal{G}_C$ can be weakly connected. 
Based on this formulation, we consider 
composite networks in which each subnetwork's outputs have been connected
to some inputs, which results 
in the composite spreading network having no external inputs.  

Dissipativity provides a bottom-up stability test method by which
the stability of a composite system can be determined
by the interconnection structure and properties of the subsystems~\cite{moylan1978stability,vidyasagar1981input}. 
In particular, the stability of a composite system can be inferred from the dissipativity of its subsystems, and a Lyapunov function 
for the composite system
can be constructed using the storage functions of the dissipative subsystems. Therefore, we first examine the conditions under which the composite spreading network $\mathcal{G}_C$ will have a unique disease-free equilibrium, based on the analysis of composite spreading networks derived from its constituent dissipative spreading networks.

\begin{theorem} \label{Thm-Stability-Composite-Network}
Consider a composite spreading network $\mathcal{G}_C$ that is composed of~$m$ spreading networks, denoted~$\mathcal{G}^k$ for $k\in\underline m$. 
For all $k\in\underline m$,
the dynamics of $\mathcal{G}^k$ is given by network $SIS$ dynamics with external inputs in~\eqref{Eq: SIS_Input}.
The  dynamics of the composite spreading network are given in~\eqref{Eq_SIS_Composite_Network_w_M}. 
The composite spreading network $\mathcal{G}_C$ will have a unique asymptotically  stable disease-free equilibrium if 
there exist $\alpha_k > 0$ for all $k \in \underline m$ such that 
\begin{align}
\label{Eq_Negative_Def}
     \begin{bmatrix}
          M^\top & I^\top
     \end{bmatrix}        
        \Psi 
    \begin{bmatrix}
        M \\ I
    \end{bmatrix}
        < \boldsymbol{0},  
\textnormal{\ and \ }
\Psi = 
\left[
\begin{array}{c|c}
\Psi^{11} & \Psi^{12} \\
\hline
\Psi^{21} & \Psi^{22}
\end{array}
\right],
\end{align}
where~$\Psi^{11} = \boldsymbol{0}$, $\Psi^{21} = {\Psi^{12}}^\top$, and 
\begin{align*}
    \Psi^{12} &= \diag\{\frac{1}{2} \alpha_1 {C^1}^\top,\dots,\frac{1}{2}\alpha_m {C^m}^\top\}, \\
      \  \Psi^{22} &=
    \diag\{\alpha^1(-\Gamma^1+B^1),\dots,\alpha_m (-\Gamma^m+B^m)\}.
\end{align*}
\end{theorem}
\emph{Proof:} See Appendix~\ref{app:t2proof}. \hfill $\blacksquare$

Theorem~\ref{Thm-Stability-Composite-Network} leverages the relationship between storage functions and supply rate functions for $m$ spreading networks in order to construct a Lyapunov function that examines the spreading behavior of the composite network. Specifically, it uses dissipativity analysis to ensure that the composite spreading network will converge to a unique disease-free equilibrium. Consequently, we are able to examine the properties of the disease-free equilibrium by solving a linear matrix inequality.

\begin{remark}
  We study how the composition of the $m$ constituent spreading networks affects the spreading behavior of the entire composite network, $\mathcal{G}_C$. According to Theorem~\ref{Thm-Stability-Composite-Network}, the behavior of the composite network is influenced by two factors: 1) the composition matrix $M$ and 2) the input matrix $C^k$ for $k \in \underline{m}$. The composition matrix $M$ is a binary matrix that captures only the network topology, such as whether there is traffic between cities or whether two individuals have interactions. In contrast, the input matrix $C^k$ defines the strength of these interactions, such as the volume of traffic between cities or the duration of contact between individuals. 
  Typical restrictions during an epidemic do not modify the underlying
  topology of a network, e.g., highways are not dismantled, but
  they do affect the use of the network, e.g., travel restrictions
  may be put in place to restrict the use of highways. Modeling
  these interventions, we consider problems in which~$M$
  is fixed and in which 
  we can modify the strength of the coupling by scaling the entries in $C^k$ to influence the behavior of the composite spreading network.
\end{remark}

Consider another composite spreading network $\hat{\mathcal{G}}_C$, 
which is constructed by composing the exact same $m$ constituent spreading networks $\mathcal{G}^k$ for all $k \in \underline{m}$. Compared to $\mathcal{G}_C$, we obtain $\hat{\mathcal{G}}_C$ by scaling the coupling strength from network $\mathcal{G}^l$ to network $\mathcal{G}^k$ by a factor of $\theta_k \in \mathbb{R}_{>0}$. As a result, the input matrix for $\mathcal{G}^k$ becomes ${\hat{C}^k = \theta_k C^k}$, where $C^k$ is the original, unscaled input transmission matrix of the $k^{th}$ spreading network as defined in Definition~\ref{Def_Input_Transmission_Matrix}. The dynamics of the composite spreading network $\hat{\mathcal{G}}_C$, denoted as
$\dot{x} = f_{\hat{C}}(x)$, are given by
\begin{multline} \label{Eq_SIS_Composite_Network_G_theta}
  \dot {x} (t) = \big({-\Gamma} + B - \diag(x)B\big) x(t) \\
   + \big(\hat{C} - \diag(x)\hat{C}\big)Mx(t),
\end{multline}
where $\hat{C} = \diag\{\hat{C}^1,\dots, \hat{C}^m\}$ and 
the other symbols are the same as in~\eqref{Eq_SIS_Composite_Network_w_M}. 


\begin{corollary} \label{cor_G_hat_C}
Suppose that we can find scaling parameters
$\alpha_k > 0$ for all $k\in\underline m$ such 
that Theorem~\ref{Thm-Stability-Composite-Network} holds for 
$\mathcal{G}_C$. Then the composite spreading network $\hat{\mathcal{G}}_C$ has the globally asymptomatic stable disease-free equilibrium if $\theta_k\leq \alpha_k$
for all $k\in \underline m$.
\end{corollary}
\emph{Proof: } See Appendix~\ref{app:cor1proof}. \hfill $\blacksquare$

Theorem~\ref{Thm-Stability-Composite-Network} and Corollary~\ref{cor_G_hat_C} demonstrate that the spreading behavior of composite spreading networks can be analyzed through their constituent spreading networks from a dissipative perspective. These results are quite general, as the $m$ constituent spreading networks in~\eqref{Eq: SIS_Input} are always dissipative.
As a result, a wide range of interconnection patterns can be immediately characterized, allowing for the analysis of large-scale spreading networks by decomposing them into smaller, strongly connected networks. This decomposition simplifies the understanding of network dynamics at smaller scales. Additionally, we can formulate optimal resource allocation problems where the solutions lie within these interconnection patterns.
A key advantage of this approach over existing work 
is that this approach relaxes the assumption that the composite spreading network is strongly connected. Together, Theorem~\ref{Thm-Stability-Composite-Network} and Corollary~\ref{cor_G_hat_C} solve Problem~\ref{Problem-Comp-Dissipativity}.




\subsection{Graph Analysis of Composite Spreading Networks}
The composite spreading network $\mathcal{G}_C$ follows the 
network $SIS$ dynamics in~\eqref{Eq: SIS_Cpct}.
In this section, we leverage graph-theoretical analysis to study the conditions under which~\eqref{Eq_Negative_Def} in Theorem~\ref{Thm-Stability-Composite-Network} either has a solution or does not. We use these results to further demonstrate how dissipativity analysis can be applied to study the spreading dynamics of the composite network through its constituent spreading networks and to guide policy-making in pandemic mitigation efforts.

As noted above, the composite spreading network $\mathcal{G}_C$ is always at least weakly connected and is composed of $m$  constituent spreading networks $\mathcal{G}^k$, where $k \in \underline{m}$.
Weakly connected graphs can always be decomposed into strongly connected subgraphs using established algorithms such as Tarjan's algorithm~\cite{tarjan1985efficient}. 
And we can investigate the existence of a solution for~\eqref{Eq_Negative_Def} in Theorem~\ref{Thm-Stability-Composite-Network} by applying established results for strongly connected spreading networks from Proposition~\ref{Prop: SIS_Equi}
to the strongly connected components. 

We use an alternative representation for the dynamics $\dot{x} = f_c(x)$ of the composite spreading network $\mathcal{G}_C$ as given in~\eqref{Eq_SIS_Composite_Network_w_M}. Consider
\begin{equation}
\label{Eq: SIS_Composite_Network}
  \dot {x} (t) = \big(-\Gamma + B_C - \diag(x)B_C\big) x(t),
\end{equation}
where $x=[{x^1}^\top, \cdots, {x^n}^\top]^\top \in\mathbb R^{\mathbf {n}}$, and $\Gamma$ is from~\eqref{Eq_SIS_Composite_Network_w_M}.
We represent the transmission matrix in~\eqref{Eq: SIS_Composite_Network} as
\begin{equation}
\label{eq_trans_matrix_comp}
B_C = \left[
\begin{array}{cccc}
B^1 & \cdots & \cdots & B^{1m} \\
\vdots & \ddots & B^{ij} & \vdots \\
\vdots & B^{ji} & \ddots & \vdots \\
B^{m1} & \cdots & \cdots & B^m \\
\end{array}
\right],
\end{equation} 
where the diagonal blocks are given by $B^k$ in~\eqref{Eq: SIS_Cpct}. 
The off-diagonal block $B^{ij}$ captures the interactions from the $j^{th}$ spreading network to the $i^{th}$ spreading network for all distinct $i,j\in\underline m$. We interpret~\eqref{Eq: SIS_Composite_Network} as a way to view a composite spreading network through its $m$ constituent spreading networks $\mathcal{G}^k$. In contrast,~\eqref{Eq_SIS_Composite_Network_w_M} represents the process of combining these $m$ strongly connected networks $\mathcal{G}^k$ into a single composite spreading network, where $k \in \underline{m}$. One can check that~\eqref{Eq_SIS_Composite_Network_w_M} and~\eqref{Eq: SIS_Composite_Network} indeed represent the same dynamics of~$\mathcal{G}_C$. 
\begin{theorem} \label{Thm_Composite_Net_Individual}
Equation~\eqref{Eq_Negative_Def} in Theorem~\ref{Thm-Stability-Composite-Network} holds only if each of the constituent spreading networks $\{\mathcal{G}^k\}_{k=1}^m$ that comprise the composite spreading network $\mathcal{G}_C$ have a globally asymptotically stable disease-free equilibrium at $x^k=\boldsymbol{0}$.
\end{theorem}
\emph{Proof: } See Appendix~\ref{app:t3proof}. \hfill $\blacksquare$

The following corollary follows immediately 
from Theorem~\ref{Thm_Composite_Net_Individual}, and
its proof is omitted. 

\begin{corollary} \label{cor_comp_network_1}
If there exists at least one  strongly connected constituent spreading network $\mathcal{G}^k$, for $k \in \underline{m}$ that has an unstable disease-free equilibrium, then the weakly connected composite spreading network $\mathcal{G}_C$ cannot reach the disease-free equilibrium.
\end{corollary}

\begin{remark}
Theorem~\ref{Thm_Composite_Net_Individual} and Corollary~\ref{cor_comp_network_1} together demonstrate that suppressing the disease spread across the entire composite network is impossible unless the spread within every strongly connected constituent network $\mathcal{G}^k$ is also mitigated. For instance, in real-world epidemic scenarios, lifting lockdown restrictions between two isolated regions during a pandemic is only reasonable if the spread in both regions has already been sufficiently suppressed. 
\end{remark}

Theorem~\ref{Thm_Composite_Net_Individual} applies to weakly connected composite spreading networks, even though the constituent networks are assumed to be strongly connected. 
The generality of the theorem and corollary lies in the fact that any weakly connected graph can be decomposed into its strongly connected components. This decomposition provides a framework for analyzing the overall spreading behavior of a composite network by examining the dynamics of its strongly connected components with smaller sizes.

Theorem~\ref{Thm_Composite_Net_Individual} and Corollary~\ref{cor_comp_network_1} show that, regardless of the connectivity type of the composite spreading network $\mathcal{G}_C$, in order to ensure that the composite spreading network converges to the disease-free equilibrium, i.e., for~\eqref{Eq_Negative_Def} in Theorem~\ref{Thm-Stability-Composite-Network} to have a solution, we must ensure that the spreading in its constituent spreading networks is fading away. Based on this result, we propose the following method to change the transmission coupling strength between distinct spreading subnetworks $\mathcal{G}^k$ and $\mathcal{G}^l$ 
to suppress the spreading process in 
the composite spreading network $\mathcal{G}_C$.

Again, for the composite spreading network 
$\mathcal{G}_C$ defined in~\eqref{Eq_SIS_Composite_Network_w_M}, 
consider another composite spreading network $\hat{\mathcal{G}}_C$ defined in~\eqref{Eq_SIS_Composite_Network_G_theta}, where the input matrix of the constituent spreading network $\mathcal{G}^k$ is given by  $\theta_k C^k$, with $\theta_k\in(0,1)$, for all $k\in\underline m$.
\begin{corollary}
\label{cor:stability_new_composite_scale}
The composite spreading network $\mathcal{G}_C$ is unstable at the disease-free equilibrium $x=\mathbf{0}$, 
but all constituent spreading networks $\mathcal{G}^k$ are asymptotically stable at their respective disease-free equilibria 
 $x^k=\mathbf{0}$. If we can find a group of $\alpha_k>0$  such that Theorem~\ref{Thm-Stability-Composite-Network} holds for 
 $\hat{\mathcal{G}}_C$, then the composite spreading network $\hat{\mathcal{G}}_C$ has the stable disease-free equilibrium.
\end{corollary}
The corollary follows directly from Theorem~\ref{Thm-Stability-Composite-Network} 
and Corollary~\ref{cor_G_hat_C}. Therefore, its proof is omitted. Corollaries~\ref{cor_G_hat_C}-~\ref{cor:stability_new_composite_scale} provide a method for designing the interconnection weights between the constituent spreading networks. By finding appropriate scaling factors $\theta_k$, $k\in\underline{m}$, these corollaries help determine the strength of the connections between the networks, ensuring that the spreading process of the composite network converges to the disease-free equilibrium. As discussed, the interconnection weights can represent various factors, such as the volume of traffic between cities or the duration of contact between individuals.
Thus, Theorem~\ref{Thm_Composite_Net_Individual} 
along with Corollaries~\ref{cor_comp_network_1} and ~\ref{cor:stability_new_composite_scale} further solve Problem~\ref{Problem-Application}.

\section{Simulation and Applications} \label{sec_application}
In this section, 
we use a real-world contact network to illustrate how to use this work to guide decision-making in pandemic mitigation scenarios.
Consider a contact tracing network that captures the mixing patterns of children in school environments, specifically from a study conducted in a French school involving children aged $6$-$12$~\cite{stehle2011high}. This network, shown in Figure~\ref{fig:French_Contacting_Network}, consists of $228$ nodes representing the students and $5,539$ edges representing their pairwise interactions over one day. Each student belongs to a specific class at the school, which is indicated by the color of the node. 
The edges between students represent the contact duration, with thicker edges indicating longer cumulative contact time throughout the day. It is observed that individuals within the same class engage in more frequent and prolonged interactions compared to those from different classes.

We consider disease-spreading through airborne transmission, such as influenza (flu). 
The transmission rate of seasonal flu is closely related to the contact time between two individuals, as longer exposure typically increases the likelihood of transmission~\cite{lakdawala2012ongoing}. 
Therefore, our key question is: during a flu season, without closing the school, how much should we intervene to reduce the overall interaction time between different classes in order to prevent an outbreak at the school?

We use several disease-spreading parameters that can be determined through public health research. We model the infectious period for influenza as $0.5 - 1.7$ days~\cite{cori2012estimating}. We can then construct a diagonal recovery matrix $\Gamma \in \mathbb R^{228\times 228}$, where 
the average recovery rate of each individual ($[\Gamma]_{ii}$) is  uniformly randomly sampled from the interval $[\frac{1}{1.7}, \frac{1}{0.5}]$.

 We use the result on 
 the role of heterogeneity in contact timing and duration in network models of influenza spread in schools
 from~\cite{toth2015role} to assume that the basic reproduction number  $R_0 = 1.05$. We use the network $SIS$ model from~\eqref{Eq: SIS_Composite_Network} to model flu.
 Longer cumulative contacting time will result in a higher transmission rate, and we therefore construct the transmission matrix $B_C = \theta \times B_T$. The matrix $B_T \in {\mathbb R}_{\geq0}^{5539\times 5539}$ is a symmetric weighted matrix of the undirected contacting network in Figure~\ref{fig:French_Contacting_Network}. The parameter $\theta$ scales the contacting time to the transmission rate. According to Proposition~\ref{Prop: SIS_Equi}, we have $\rho(\theta \Gamma^{-1}B_T) = R_0 = 1.05$, and we obtain~$\theta = 6.5376 \times 10^{-5}$. 

 Given that the basic reproduction number $R_0 = 1.05$ at school, according to Proposition~\ref{Prop: SIS_Equi}, an outbreak is likely. To address this mitigation challenge, we apply Theorems~\ref{Thm-Stability-Composite-Network} and~\ref{Thm_Composite_Net_Individual}.
 First, we compute and plot the basic reproduction numbers of the ten network $SIS$ spreading models, corresponding to the ten classes, as shown in Figure~\ref{fig:R_0_Ten_Classes}. The plot reveals that all basic reproduction numbers are less than one. Based on Corollary~\ref{cor_comp_network_1}, to prevent a potential outbreak, it is unnecessary to prioritize isolating any classes.  Next, we assess the level of restrictions to be implemented on the ten interconnected classes, which are numbered from Classes $1$ to $10$.
 


We use Theorem~\ref{Thm-Stability-Composite-Network}, Corollaries~\ref{cor_G_hat_C} and~\ref{cor:stability_new_composite_scale} to compute the level of restriction necessary to prevent an outbreak. According to~\eqref{Eq: SIS_Input} and Definition~\ref{Def_Input_Transmission_Matrix}, we construct the input matrices for the ten classes. The input matrix $C^1\in\mathbb{R}_{\geq 0}^{23\times 103}$ for Class~$1$ is shown in Figure~\ref{fig:Input_Matrix_7},  which shows that $23$ students in this class interact with $103$ students from the other connected classes at school. The color gradient of the heatmap represents the transmission rates between students. 
 
 We further visualize the matrix $\Psi\in\mathbb R^{1480\times 1480}$ from Theorem~\ref{Thm-Stability-Composite-Network} and the composition matrix $M\in\mathbb R^{1480\times 228}$ from Definition~\ref{Def_Comp_M} in Figures~\ref{fig:Diss_Matrix_Phi} and~\ref{fig:Binary_Matrix_M}, respectively. Figure~\ref{fig:Diss_Matrix_Phi} illustrates the block-partitioning structure of $\Psi$, where the diagonal block $\Psi_{22}$ is determined by the transmission and recovery matrices of each class, and the off-diagonal block matrices $\Psi_{12}=\Psi^\top_{21}$ capture the input matrices of each class, scaled by $\frac{1}{2}$. Further, Figure~\ref{fig:Binary_Matrix_M} depicts the composition matrix $M$ whose entries are $M_{ij}\in\{0,1\}$ for $i\in \underline{228}$ and $j\in\underline {1480} $.

 Last, based on Corollary~\ref{cor:stability_new_composite_scale}, we find that $\zeta = 0.798$ is the required scaling factor for adjusting the transmission coupling strength between the constituent networks to stabilize the overall composite spreading network. This value ensures that by decreasing the average interaction time between any pair of classes to less than $79\%$ of the original interaction time, there will be no outbreak among these ten connected classes.
 

\begin{figure}
    \centering
    \includegraphics[trim={1cm 2cm 2cm 1cm}, clip, width=\columnwidth]{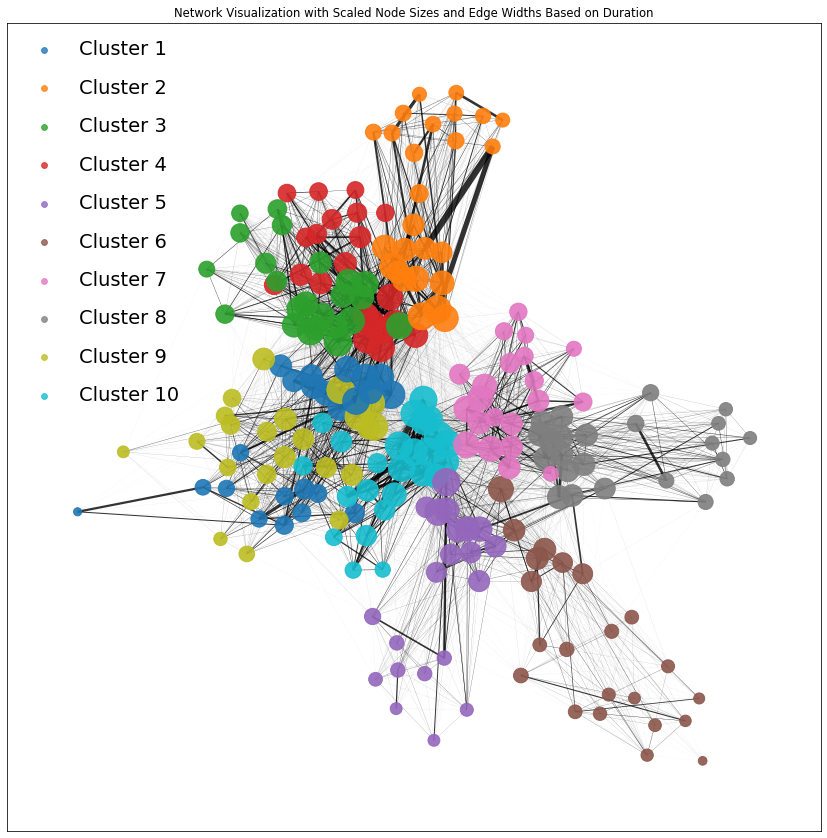}
    \caption{A undirected contacting network in a French school involving children aged $6$-$12$~\cite{stehle2011high}. This network consists of $228$ nodes representing the students and $5,539$ edges representing their pairwise interactions over one day. Each student belongs to a specific class (indexed from $1$ to $10$) at the school, which is indicated by the color of the node. 
The edges between students represent the contact duration, with thicker edges indicating longer cumulative contact time throughout the day.
    }    \label{fig:French_Contacting_Network}
\end{figure}

\begin{figure}
    \centering
    \includegraphics[width=1\linewidth]{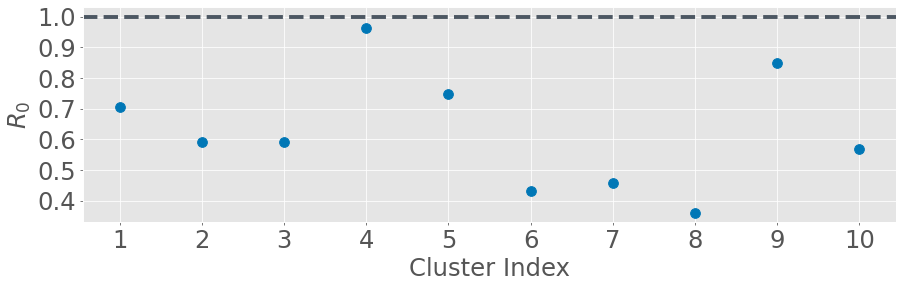}
    \caption{$R_0$ of the ten classes. We present the basic reproduction numbers of the ten classes. The basic reproduction number of the $k^{th}$ spreading network is computed as $\rho((\Gamma^k)^{-1}B^k)$, $k\in\underline{10}$. We observe that the basic reproduction numbers for all classes are less than one, indicating that it is unnecessary to isolate any classes to avoid an outbreak.}
    \label{fig:R_0_Ten_Classes}
\end{figure}

\begin{figure}
    \centering
    \includegraphics[width=1\linewidth]{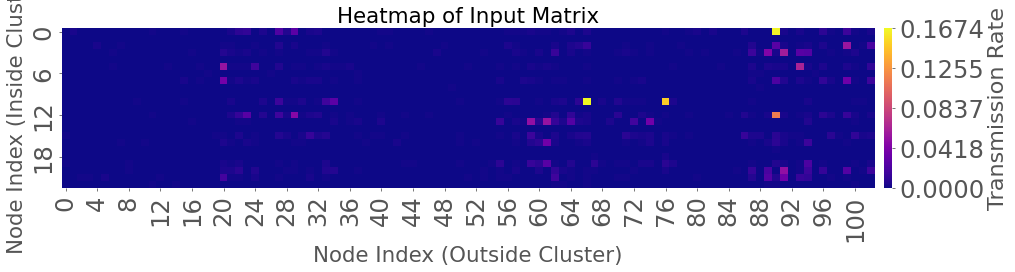}
    \caption{
    The input matrix $C^1\in\mathbb R^{23\times 103}$ for Class $1$.  The input matrix shows that $23$ students in this class interact with $103$ students from the other connected classes. The color gradient of the heatmap represents the transmission rates between students.}
    \label{fig:Input_Matrix_7}
\end{figure}

\begin{figure}
    \centering
    \includegraphics[width=1\linewidth]{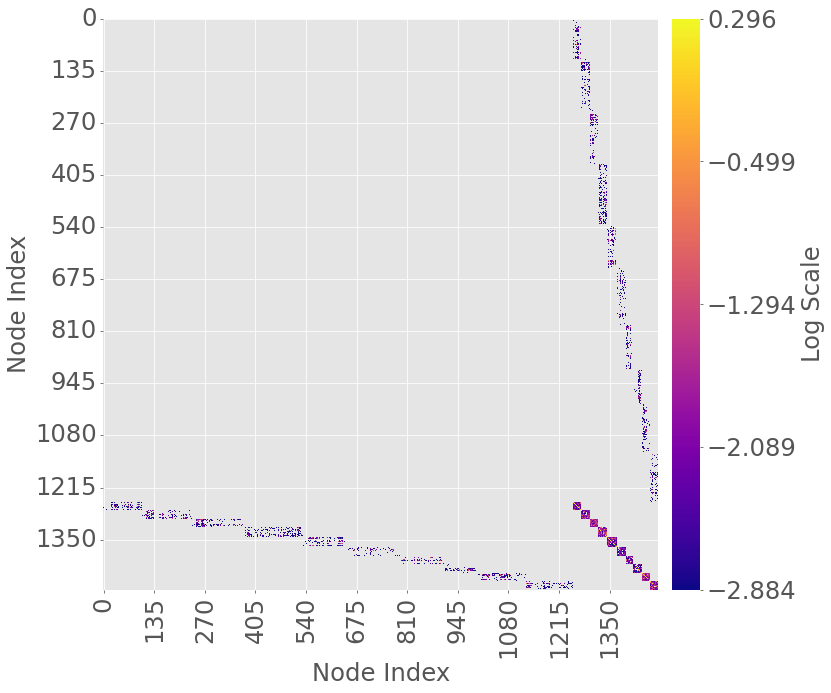}
    \caption{The matrix $\Psi\in\mathbb R^{1480\times 1480}$ defined in~\eqref{Eq_Negative_Def}. We use the log scale of the entries to better plot them.}
    \label{fig:Diss_Matrix_Phi}
\end{figure}

\begin{figure}
    \centering
    \includegraphics[width=1\linewidth]{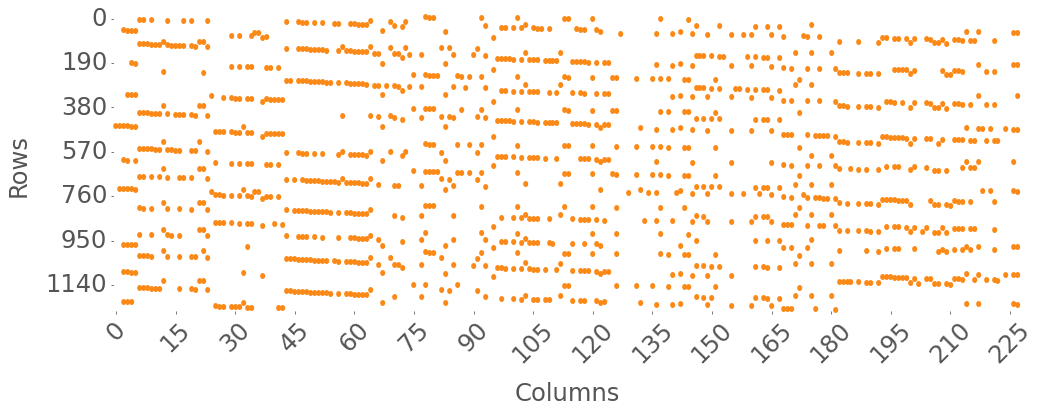}
    \caption{The binary composition matrix $M\in\mathbb R^{1480\times 228}$ defined in Definition~\ref{Def_Comp_M}.}
    \label{fig:Binary_Matrix_M}
\end{figure}
\section{Conclusion}
\label{sec_conclusion}
In this work, we developed a comprehensive model for spreading networks with defined inputs and outputs, as well as an operation for composing such networks. Using dissipativity theory, we introduced storage and supply rate functions to characterize the input-output behavior of network spreading dynamics. Building on these analyses, we established conditions under which disease spreading in a composite network $SIS$ model is guaranteed to decay to zero. Additionally, we derived explicit conditions under which such guarantees cannot be ensured.

To demonstrate the practical implications of our results, we applied our framework to simulations of influenza spread in a primary school setting. Using the developed analysis, we quantified the effectiveness of intervention strategies by showing that reducing the average interaction time between any pair of classes to a specific fraction of the original interaction time can effectively prevent an outbreak on campus.

Future work will explore the impact of control inputs in one constituent spreading network on the behavior of other constituent networks, as well as on the overall composite spreading network. Currently, our framework assumes that spreading networks evolve on the same time scale. Extending the framework to accommodate spreading dynamics occurring at different resolutions is a promising direction for further research.
\bibliographystyle{IEEEtran}
\bibliography{main}

\newpage
\appendix
\subsection{Linear Algebra Preliminaries}


\begin{lemma}
\baike{\cite[Sec. 2.1 and Lemma 2.3]{varga2009matrix_book}} 
\label{lem:irr_M} Suppose that $M$ is an irreducible Metzler matrix. Then, $s\left(M\right)$ is a simple eigenvalue of $M$ and
there exists a unique (up to scalar multiple) vector $x\gg \boldsymbol{0}$ such
that $Mx=s\left(M\right)x$. Let $z>\boldsymbol{0}$ be a vector in $\mathbb R^{n}$. If
$Mz<\lambda z$, then $s(M)<\lambda$. If $Mz=\lambda z$, then $s(M)=\lambda$.
If $Mz>\lambda z$, then $s(M)>\lambda$.
\end{lemma}

\begin{lemma}
\label{lem:irr_spe}\cite[Prop. 1]{bivirus}
Suppose that $\varLambda$ is a negative diagonal matrix in $\mathbb R^{n\times n}$
and $N$ is an irreducible nonnegative matrix in $\mathbb R^{n\times n}$.
Let $M=\varLambda+N$. Then, $s(M)<0$ if and only if $\rho(-\varLambda^{-1}N)<1$,
$s(M)=0$ if and only if $\rho(-\varLambda^{-1}N)=0$, and $s(M)>0$
if and only if $\rho(-\varLambda^{-1}N)~>~1$. 
\end{lemma}

\subsection{Proof of Theorem~\ref{Thm: SIS_Dissipativity}} \label{app:t1proof}
First we discuss the properties of the storage function. 
The storage function for the $k^{th}$ spreading network in~\eqref{Eq: SIS_Input} is defined as the $l^2$-norm of its infected state vector $x^k \in \mathbb{R}^{n_k}$, where $k \in \underline{m}$. Consequently, the storage function $V^k(x^k)$ is continuously differentiable with respect to $x^k$ for all $k \in \underline{m}$. Under the condition that the infected state $x^k \in [0, 1]^{n_k}$, $V^k(x^k)$ is positive definite for the spreading dynamics and satisfies $V^k(x^k) = 0$ only when $x^k = \boldsymbol{0}$.
Thus, we can leverage Definition~\ref{Def: Dissipativity} to show the network spreading dynamics is dissipative with based on this storage function.

In order to prove that the $k^{th}$ network $SIS$ dynamics in~\eqref{Eq: SIS_Input} is dissipative with respect to the supply
rate function
\begin{align*}
   S^k(u^k,y^k) = y^k(t)^\top C^k u^k(t) + y^k(t)^\top(-\Gamma^k+ B^k) y^k(t), 
\end{align*}
we need to show that the
storage function $V^k(x^k) = \frac{1}{2}\|x^k\|^2$
introduced in Definition~\ref{Def_Storage_Function} satisfies
\begin{align*}
 \nabla V^k(x^k)^\top f^k(x^k,u^k)\leq S^k(u^k,y^k).   
\end{align*}

To compare the gradient of the
storage function along the trajectory of the system dynamics, and the
supply rate function, we have that
\begin{align*}
&\nabla V^k(x^k)^\top f^k(x^k,u^k)- S^k(u^k,y^k)\\
&={x^k}^\top(-\Gamma^k+B^k)x^k - {x^k}^\top \diag(x^k)B^kx^k(t) 
\\
& \ \ \ \ + {x^k}^\top (C^k - \diag(x^k)C^k) u^k (t)- {y^k}^\top C^k u^k(t) 
\\
&\ \ \ \ - {y^k}^\top (-\Gamma^k+B^k) y^k.
\end{align*}
According to~\eqref{eq:output}, the output vector of the $k^{th}$ spreading network equals its state vector, i.e., $y^k(t) = x^k(t)$, for all $k\in\underline m$. By substituting this condition into the equation above, we have that
\begin{align}
\label{Eq: Them_1_Proof}
- x^k{^\top} \diag(x^k)B^kx^k
 - {x^k}^\top\diag(x^k)C^k u^k\leq 0.
\end{align}
Hence, according to Definition~\ref{Def: Dissipativity}, for all $k\in\underline m$, the $k^{th}$ spreading network is dissipative with respect to the storage function and the supply rate function defined in~Definitions~\ref{Def_Storage_Function} and~\ref{Def_Supply_Rate}, respectively.
\hfill $\blacksquare$

\subsection{Proof of Lemma~\ref{Lem: supply_rate}} \label{app:l1proof}
 We study the first term  ${y^k}^\top C^k u^k$ in $S^k(u^k,y^k)$ for $k\in\underline m$.
    We use the input vector $u^k$ and the output vector $y^k$ to represent the infected proportions, such that $u^k\in[0,1]^{p_k}$ and $y^k\in[0,1]^{n_k}$ for $k\in\underline m$. Based on the definition of the input transmission matrix in Definition~\ref{Def_Input_Transmission_Matrix}, 
    the input matrix $C^k$ is a non-negative matrix. Therefore, ${y^k}^\top C^k u^k\geq 0$ for all $y^k$ and $u^k$ and $k\in\underline m$.

    Then, we analyze the second term ${y^k}^\top (-\Gamma^k +B^k) y^k$ in $S^k(u^k,y^k)$. To ensure that the supply rate is negative, we must satisfy ${y^k}^\top(-\Gamma^k + B^k){y^k} < 0$. 
   Under the condition that the  spreading network $\mathcal{G}^{k}$ is undirected, 
    we obtain that the transmission matrix 
    $B^k$ and the matrix $-\Gamma^k +B^k$ are symmetric.
    Therefore, the quadratic form ${y^k}^\top(-\Gamma^k + B^k){y^k} < 0$ indicates that $\sigma(-\Gamma^k +B^k)<\mathbf{0}$. 
    According to~Proposition~\ref{Prop: SIS_Equi}, 
    the spreading network $\mathcal{G}^k$ must have the globally stable disease-free equilibrium at $x^k = \mathbf{0}$ for $k\in\underline m$. \hfill $\blacksquare$

\subsection{Proof of Theorem~\ref{Thm-Stability-Composite-Network}} \label{app:t2proof}
Consider the network $SIS$ spreading dynamics of $\mathcal{G}^k$ given by~\eqref{eq:dyn} and~\eqref{eq:output}. 
According to Definition~\ref{Def_Storage_Function}, we define
the storage function of the $k^{th}$ spreading network as $V^k(x^k) = \frac{1}{2}\|x^k\|$ with $x^k\in\mathbb R^{n^k}$. Hence, 
we consider a weighted sum of the storage functions 
that serves as a Lyapunov function for the composite spreading network $\mathcal{G}_C$:
\begin{align*}
 V_C(x) = \sum_{k=1}^m \alpha_k V^k(x^k),\textnormal{ \ where \ } \alpha_k> 0, \,\, k\in\underline m.
\end{align*}
According to the definition, $V_C(x)\geq 0$ for all $x\in\mathbb{R}^\mathbf{n}$, since $V^k(x^k)\geq 0$ for all $x^k\in[0,1]^{n_k}$, and $\alpha_k> 0$ for all $k\in\underline m$.

According to Theorem~\ref{Thm: SIS_Dissipativity}, the $k^{th}$ $SIS$ constituent spreading network is dissipative with respect to the positive definite, continuously differentiable storage function $V^k(x^k)=\frac{1}{2}\|x^k\|^2$, and the quadratic supply rate function $S^k(u^k,y^k) = y^k(t)^\top C^k u^k(t) + y^k(t)^\top(-\Gamma^k+ B^k) y^k(t)$ for $k\in\underline m$. 

Under the fact that 1) the coupling strength given by the  composition matrix $M$ in Definition~\ref{Def_Comp_M} is a binary matrix, and 2) the input vector $u^k$ for the $k^{th}$ network $SIS$ dynamics in~\eqref{Eq: SIS_Input} contains no states from itself, we can verify that 
\begin{align*}
 \nabla V_C(x)^\top f_c(x) = \sum_{k=1}^m \alpha_k \nabla V^k(x^k)^\top f^k(x^k,u^k).
\end{align*}
Further, based on Theorem~\ref{Thm: SIS_Dissipativity}, we have that
\begin{align*}
\nonumber
&\sum_{k=1}^m \alpha_k \nabla V^k(x^k)^\top f^k(x^k,u^k)
\leq\sum_{k=1}^m \alpha^k S^k(u^k,g^k)\\
&=\sum_{k=1}^m \alpha^k
\begin{bmatrix}
u^k\\
y^k
\end{bmatrix}^\top 
\begin{bmatrix}
0 & \frac{1}{2}{C^k}^{\top}\\
\frac{1}{2}C^k & -\Gamma^k+ B^k
\end{bmatrix}
\begin{bmatrix}
u^k \\ y^k
\end{bmatrix}
= \begin{bmatrix}
u\\
y
\end{bmatrix}^\top 
\Psi
\begin{bmatrix}
u \\ y
\end{bmatrix}\\
& = \begin{bmatrix}
My\\
y
\end{bmatrix}^\top 
\Psi
\begin{bmatrix}
My \\ y
\end{bmatrix}
= y^\top\begin{bmatrix}
M\\
I
\end{bmatrix}^\top 
\Psi
\begin{bmatrix}
M \\ I
\end{bmatrix}y\\
&=x^\top
\begin{bmatrix}
M\\
I
\end{bmatrix}^\top 
\Psi
\begin{bmatrix}
M \\ I
\end{bmatrix}x,
\end{align*}
where we use the condition in Definition~\ref{Def_Comp_M} that $u = M y$, and we can observe all states, i.e., $y = x$. Further, $\Psi$ is defined in the theorem.

Therefore, under the condition that $x\in[0,1]^{\mathbf{n}}$, 
if we can find $\alpha^k> 0$ for all $k\in\underline m$ such that 
\begin{align*}
\begin{bmatrix}
M\\
I
\end{bmatrix}^\top \Psi \begin{bmatrix}
M \\ I
\end{bmatrix}< \boldsymbol{0},  
\end{align*}
then we will ensure that $\nabla V_C(x)^\top f_c(x)<0$ for all $x\in[0,1]^{\mathbf{n}}$. 
Further, according to~\eqref{eq:output}, the state-out function is $y^k(x^k, u^k) = x^k$ for all $k\in\underline m$. Therefore, we have
$y^k(x^k, 0)=\boldsymbol{0}$ if and only if $x^k=\boldsymbol{0}$, for all $k\in\underline m$.
Thus, according to the discussion from~\cite[Proposition~2.1]{arcak2016networks}, the disease-free equilibrium at
$x=\boldsymbol{0}$ is asymptotically stable for $x\in [0,1]^\mathbf{n}$.
\hfill $\blacksquare$

\subsection{Proof of Corollary~\ref{cor_G_hat_C}} \label{app:cor1proof}
Consider that we can find a collection of $\alpha_k > 0$ for all $k \in \underline{m}$ such that Theorem~\ref{Thm-Stability-Composite-Network} holds for the composite spreading network $\mathcal{G}_C$. Then, for $\theta_k \in (0, 1]$ for all $k \in \underline{m}$, we have that
\begin{align*}
     \begin{bmatrix}
          M & I
     \end{bmatrix}        
        \hat\Psi 
    \begin{bmatrix}
        M \\ I
    \end{bmatrix}
        < \boldsymbol{0},  
\textnormal{\ and \ }
\hat\Psi = 
\left[
\begin{array}{c|c}
\hat{\Psi}^{11} & \hat{\Psi}^{12} \\
\hline
\hat{\Psi}^{21} & \hat{\Psi}^{22}
\end{array}
\right],
\end{align*}
where 
\begin{align*}
    \hat{\Psi}^{12} &= 
    \begin{bmatrix}
       \frac{1}{2} \alpha_1 \theta_1 {C^1}^\top &   &\\
        & \ddots & \\
        &  &   \frac{1}{2}\alpha_m\theta_m{C^m}^\top\\
    \end{bmatrix},\\
\hat{\Psi}^{11} &=  \Psi^{11},
\hat{\Psi}^{21} = {\Psi^{12}}^\top, \textnormal{\ and \ } \hat{\Psi}^{22} ={\Psi}^{22}.
\end{align*}
We utilize the fact that multiplying scalars ${\theta^k \in (0, 1]}$ to the linear matrix inequality in Theorem~\ref{Thm-Stability-Composite-Network} preserves the inequality. By substituting $\hat{C}^k = \theta_k C^k$ into $\hat{\Psi}^{12}$, we can apply Theorem~\ref{Thm-Stability-Composite-Network} to demonstrate that the composite spreading network $\hat{\mathcal{G}}_C$ has a globally asymptotically stable disease-free equilibrium. \hfill $\blacksquare$

\subsection{Proof of Theorem~\ref{Thm_Composite_Net_Individual}} \label{app:t3proof}
Without loss of generality, consider the condition that there exists one strongly connected constituent spreading network $\mathcal{G}^i$ 
with~$i \in \underline{m}$
such that its disease-free equilibrium is unstable. According to Proposition~\ref{Prop: SIS_Equi}, we have that $\sigma({\Gamma^i}-B^i)> 0$ and 
$\sigma({\Gamma^j}-B^j)\leq 0$ for
$i\in\underline m$ and $j\in\underline m\setminus\{i\}$. 

Based on the fact that the eigenvalues of a block diagonal matrix are the union of the spectra of its individual block matrices, we conclude that
\begin{align*}
&\sigma(B-\Gamma)=\sigma(\diag\{B^1,\cdots,B^m\}-\diag\{\Gamma^1,\cdots,\Gamma^m\})\\
&=\sigma(\diag\{B^1-\Gamma^1,\cdots,B^m-\Gamma^m\}\\
&=\max_{k\in\underline m}\{\sigma(B^k-\Gamma^k)\}
= \sigma(B^i-\Gamma^i) > 0.
\end{align*}

Then, we compare $\sigma(B-\Gamma)$ and $\sigma(B_C-\Gamma)$. Based on the dynamics of the composite spreading network $\mathcal{G}_C$ in~\eqref{Eq: SIS_Composite_Network}, we obtain the transmission matrix $B_C$ in~\eqref{eq_trans_matrix_comp} by adding off-diagonal entries to $B=\diag\{B^1,\cdots,B^m\}$, denoted by $B^{ij}$ for $i,j \in \underline{m}$, $i \neq j$. Furthermore, given that transmission rates can only take positive values, we have $B^{ij} > 0$ and $B_C > 0$ for all $i,j \in \underline{m}$ in~\eqref{Eq: SIS_Composite_Network}. Hence, both $B_C-\Gamma$ and $B-\Gamma$ are Metzler matrices.
As demonstrated in~\cite[Lemma 2]{cvetkovic2020stabilizing}, for two Metzler matrices $A, B \in \mathbb{R}^{n \times n}$, if $A > B$, then $\sigma(A) > \sigma(B)$. Thus, we have that
\begin{align*} \sigma(B_C-\Gamma) > \sigma(B-\Gamma) > 0,
\end{align*} since $B_C > \diag\{B^1,\cdots,B^m\}$.

Last, we consider the Jacobian matrix of the dynamics of the
composite weakly connected spreading network in~\eqref{Eq: SIS_Composite_Network} evaluated at the disease-free equilibrium  $x = \mathbf{0}$, given by
\begin{align*}
    J_{x=\mathbf{0}} &= (-\Gamma+B_C)+\diag(B_Cx)+\diag(x)B_C\\&=-\Gamma+B_C.
\end{align*}
Thus, the disease-free equilibrium at 
$x = \mathbf{0}$
is unstable, since $\sigma (J_{x=\mathbf{0}}) = \sigma(B_C-\Gamma)>0$. Further, we conclude that 
when there exists a strongly connected constituent spreading network $\mathcal{G}^i$ such that the disease-free equilibrium of 
$\mathcal{G}^i$
is unstable, we must have that
the weakly connected composite network is unstable at its disease-free equilibrium. Hence, according to Theorem~\ref{Thm-Stability-Composite-Network}, we cannot find any feasible solution to satisfy~\eqref{Eq_Negative_Def} to show that the composite spreading network $\mathcal{G}^k$ is stable at the disease-free equilibrium $x = \mathbf{0}$.
\hfill $\blacksquare$

\end{document}